\documentclass[conference]{IEEEtran}
\IEEEoverridecommandlockouts

\usepackage[dvipdfmx]{graphicx}

\usepackage{cite}
\usepackage{times}
\usepackage{bm}
\usepackage{braket}
\usepackage{graphicx}
\usepackage{theorem}
\usepackage{amsmath}
\usepackage{amssymb,mathrsfs}
\usepackage{latexsym}
\usepackage{caption}
\usepackage{comment,xcolor}
\usepackage{algorithmic}

\newtheorem{thm}{Theorem}
\newtheorem{lem}{Lemma}

\newtheorem{coro}{Corollary}

\theorembodyfont{\rmfamily}
{\bf}{\rm}
{\bf}{\rm}

{\itshape}{\rmfamily}
{\itshape}{\rmfamily}
\newenvironment{proof}{\noindent{\it Proof.~~}}{\medskip}

\makeatletter
\def\eqnarray{\stepcounter{equation}\let\@currentlabel=\theequation
\global\@eqnswtrue
\global\@eqcnt\z@\tabskip\@centering\let\\=\@eqncr
$$\halign to \displaywidth\bgroup\@eqnsel\hskip\@centering
  $\displaystyle\tabskip\z@{##}$&\global\@eqcnt\@ne 
  \hfil$\;{##}\;$\hfil
  &\global\@eqcnt\tw@ $\displaystyle\tabskip\z@{##}$\hfil 
   \tabskip\@centering&\llap{##}\tabskip\z@\cr}
\makeatother
\makeatletter
\def\Left#1#2\Right{\begingroup%
   \def\ts@r{\nulldelimiterspace=0pt \mathsurround=0pt}%
   \let\@hat=#1%
   \def\sht@im{#2}%
   \def\@t{{\mathchoice{\def\@fen{\displaystyle}\k@fel}%
          {\def\@fen{\textstyle}\k@fel}%
          {\def\@fen{\scriptstyle}\k@fel}%
          {\def\@fen{\scriptscriptstyle}\k@fel}}}%
   \def\g@rin{\ts@r\left\@hat\vphantom{\sht@im}\right.}%
   \def\k@fel{\setbox0=\hbox{$\@fen\g@rin$}\hbox{%
      $\@fen \kern.3875\wd0 \copy0 \kern-.3875\wd0%
      \llap{\copy0}\kern.3875\wd0$}}%
      \def\pt@h{\mathopen\@t}\pt@h\sht@im%
      \Right}%
\def\Right#1{\let\@hat=#1%
   \def\st@m{\mathclose\@t}%
   \st@m\endgroup}
\makeatother

\makeatletter
\DeclareRobustCommand\widecheck[1]{{\mathpalette\@widecheck{#1}}}
\def\@widecheck#1#2{%
    \setbox\z@\hbox{\m@th$#1#2$}%
    \setbox\tw@\hbox{\m@th$#1%
       \widehat{%
          \vrule\@width\z@\@height\ht\z@
          \vrule\@height\z@\@width\wd\z@}$}%
    \dp\tw@-\ht\z@
    \@tempdima\ht\z@ \advance\@tempdima2\ht\tw@ \divide\@tempdima\thr@@
    \setbox\tw@\hbox{%
       \raise\@tempdima\hbox{\scalebox{1}[-1]{\lower\@tempdima\box
\tw@}}}%
    {\ooalign{\box\tw@ \cr \box\z@}}}
\makeatother

\newcommand{\wt}{\widetilde}

\newcommand{\qed}{\hspace*{\fill}$\Box$}
\newcommand{\EE}{\mathsf{E}}
\newcommand{\VV}{\mathsf{V}}
\newcommand{\PP}{\mathsf{P}}

\newcommand{\bbK}{\mathbb{K}}

\newcommand{\bbN}{\mathbb{N}}
\newcommand{\bbR}{\mathbb{R}}

\newcommand{\bbZ}{\mathbb{Z}}

\newcommand{\rd}{{\rm d}}
\newcommand{\re}{{\rm e}}

\def\BibTeX{{\rm B\kern-.05em{\sc i\kern-.025em b}\kern-.08em
    T\kern-.1667em\lower.7ex\hbox{E}\kern-.125emX}}

\begin{document}
%
%


\title{Correlation Coefficient Analysis of the \\Age of Information in Multi-Source Systems}

\author{\IEEEauthorblockN{Yukang Jiang$^*$, Kiichi Tokuyama$^*$, Yuichiro Wada$^\dagger$, Moeko Yajima$^*$}
\\
\IEEEauthorblockA{$^*$ Tokyo Institute of Technology, jykjapan008@gmail.com, \{tokuyama.k.aa, yajima.m.ad\}@m.titech.ac.jp}
\IEEEauthorblockA{$^\dagger$ FUJITSU LABORATORIES LTD. / RIKEN AIP, wada.yuichiro@fujitsu.com}
}


\maketitle

\begin{abstract}

This paper studies the age of information (AoI) on an information updating system such that multiple sources share one server to process packets of updated information.
In such systems, packets from different sources compete for the server, and thus they may  suffer from being interrupted, being backlogged, and becoming stale.
Therefore, in order to grasp structures of such systems, it is crucially important to study a metric indicating a correlation of different sources.
In this paper, we aim to analyze the correlation of AoIs on a single-server queueing system with multiple sources. As our contribution, we provide 
the closed-form expression of the \emph{correlation coefficient} of the AoIs. To this end, we first derive the Laplace-Stieltjes transform of the stationary distribution of each AoI for the multiple sources. Some nontrivial properties on the systems are revealed from our analysis results.

\end{abstract}

\begin{IEEEkeywords}
Age of information, queueing theory, single-server queues, multiple sources, correlation coefficient, stationary distribution.
\end{IEEEkeywords}

\section{Introduction}

In recent years, we can see the real-time information updating systems in many places because of the ever-increasing demand of controlling time-critical information throughout network systems.
The typical examples are monitoring systems of weather reports, vehicular status update systems that assist self-driving of cars, remote controlling of construction machinery, etc.
In such systems, various kinds of status are displayed on equipped monitors (e.g., temperature, humidity, and air pressure in weather reports; position, velocity, and acceleration in vehicular status).
When updated status of one kind is captured by a sensor, a packet is generated by the associated source of an updated information.
Thereafter, it is processed by one of servers, and then displayed on the corresponding monitor.
In addition, while all the servers are busy with processing, a newly arriving packet is backlogged and becomes outdated. 
Note that these situations occur in practice if the arrival frequency of packets is beyond the processing power of servers. 
Owing to the above properties of the systems, the information displayed on the monitor is not always up-to-date. Therefore, the freshness of the displayed information should be quantified and managed for those information updating systems.

From such backgrounds, a performance metric called the age of information (AoI) was proposed \cite{KaulYatesGrut12}. To define the AoI, let $\eta(t)$ denote the timestamp of the generation time of the information displayed on the monitor at time $t$. Then, the AoI at time $t$ is defined by
\begin{equation}\nonumber
\Delta_t:=t-\eta(t).
\end{equation} 
The AoI defined by $\Delta_t$ can indicate the freshness, because the above expression means the elapsed time from the generation of the information displayed on the monitor.  
We note that, in a system with multiple sources, the AoI is defined for each source.

The AoI in queueing systems have been studied in recent years. We here focus on previous studies which investigated the AoIs in queueing systems with multiple information sources.
Yates and Kaul provided the pioneering study in \cite{YatesKaul12}. They considered single-server queueing systems where two sources share one server to process updated information, and derived the userful expression of the average AoI of each information source. They also studied the first-come-first-served (FCFS) M/M/1 system as a special case, and derived the closed-form expression of the average AoI of each source.
Kaul and Yates \cite{KaulYates12} showed that two-source M/M/1/1 systems with preemption outperform two-source FCFS M/M/1 systems in terms of the average AoI, which triggered further studies for such without-queue systems. Najm and Telatar \cite{NajmTela18} derived the average AoI and the peak AoI of each source in two-source M/G/1 systems with preemption. The systems with general multiple information sources were investigated by Yates and Kaul \cite{YatesKaul18}. They considered the average AoI on M/M/1/1 and  M/M/1/2 systems, in both of which multiple sources share one server. In case of M/M/1/2, they utilized stochastic hybrid systems to discard waiting packets, and then they successfully reduced complexity of their analysis.

As described above, most of the works handling multiple information sources have been devoted to analyzing the average AoI of each source. To the best of our knowledge, no previous works try analytical studies about the correlation of the AoIs from different information sources. 
In order to grasp structures of the information updating systems with the multiple sources, it is crucially important to study a metric indicating a correlation of different sources in addition to the AoIs of the individual sources.
Considering  the correlation leads to  a better management of an information updating system with multiple sources.

In this paper, we study the correlation of  AoIs of  a status updating system with multiple sources, which is modeled by an M/M/1/1 queueing system with preemption. Our model is assumed to consist of multiple information sources and the corresponding monitors and to share one server. Our model described above is defined in detail in Section~\ref{sec:model}.

The contribution of this paper is as follows. We first derive the {\it Laplace-Stieltjes transform} (LST) of the stationary distribution of each AoI in our model. The LST enables to obtain the mean and the variance of the AoI.  Next, assuming that the number of sources is two, we provide the closed-form expression of the {\it correlation coefficient} of the AoIs, which is the main contribution of this paper. Furthermore, using this result, we reveal some nontrivial properties on our model.


The rest of this paper is organized as follows. In Section~\ref{sec:model}, we describe the model of our investigating information updating systems. Sections~\ref{sec:analysis1} and~\ref{sec:analysis2} present analysis results. In Section~\ref{sec:analysis1}, we derive the LST for each AoI. In Section~\ref{sec:analysis2}, assuming that the number of sources is two, we obtain the correlation coefficient of the AoIs. Numerical experiments are conducted in Section~\ref{sec:evaluation}. Finally, this paper is concluded in Section~\ref{sec:conclusion}.
%
\medskip


\section{System Model}\label{sec:model}

We consider an information updating system such that multiple sources generate information packets of updated status. Generated packets are immediately transmitted to an M/M/1/1 queueing system, which is illustrated in Fig.~\ref{fig:model}. Note that an M/M/1/1 queueing system has only one server and no buffer space.
After being processed in the server, the packets are directly sent to monitors, and then the monitors display the updated information. 
Each source has the corresponding monitor, and updated information of a source is displayed on its corresponding monitor.
The number of sources is denoted as $K\in\bbN$, and $\bbK:=\{1,2,\dots,K\}$ denotes the set of type indexes of sources. 

Packets are generated from source $k$ ($k\in\bbK$) according to the time-homogeneous Poisson process with rate $\lambda_k$. Service times of packets are assumed to be independent and identically distributed (i.i.d.) with the exponential distribution having mean $1/\mu$; that is, packets from all sources have the same service time distribution.
Besides, preemption is assumed in our system; that is, the packet which currently occupies the server will be pushed out if a new packet arrives before its service completion. 
We refer to a packet which completes its service without being pushed out as the {\it valid packet}.
Henceforth, we refer to generation times of packets as arrival times. In addition, we define $\lambda:=\sum_{k\in\bbK}\lambda_k$ as the total arrival rate of the $K$ sources.
 
\begin{figure}[!t]
\begin{center}
\includegraphics[clip,width=.8\linewidth]{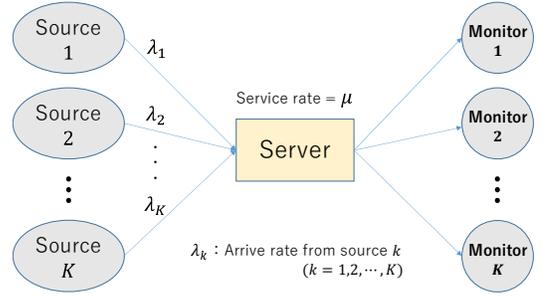}
\end{center}
\caption{Our investigating system: a packet is generated from one of the $K$ sources, processed by the server, and then displayed on the corresponding monitor. The flow is modeled according to M/M/1/1 queueing model with preemption.}
\label{fig:model}
\end{figure}
\medskip


\section{AoI for each source}\label{sec:analysis1}

In this section, we derive the LST of stationary AoI of each source. For $k\in\bbK$ and $n\in\bbZ:=\{0,\pm 1,\pm2,\dots\}$, let $\alpha_{k,n}$ denote the $n$th arrival time of the packet of source $k$, and $S_{k,n}$ denote the service time of the packet of source $k$ which arrives at $\alpha_{k,n}$. We define $A_k(t)$ as the AoI of source $k\in\bbK$ at time $t\in\bbR$. Using these notations, we have the following expression, for $k\in\bbK$ and $t\in\bbR$.
\[
A_k(t)
=t-\max_{n\in\bbZ}\left\{\alpha_{k,n};\alpha_{k,n+1}-\alpha_{k,n}>S_{k,n},t>\beta_{k,n}\right\},
\]
where $\beta_{k,n}:= \alpha_{k,n}+S_{k,n}$. Samples paths of $A_k(t)$ are illustrated in Figs.~\ref{fig:AoI1} and~\ref{fig:AoI2}. $\{(\alpha_{k,n},S_{k,n})\}_{n\in\bbZ}$ is the stationary and ergodic marked point process. Thus, we define $A_k$, $k\in\bbK$, as the random variable following the stationary distribution of $\{A_k(t)\}_{t\in\bbR}$. 

In addition, we define some notations related to valid packets. For $k\in\bbK$ and $n\in\bbZ$, we define $\alpha^*_{k,n}$ and  $S_{k,n}^*$ as the  arrival time and service time of the $n$th valid packet of source $k$, respectively. We also define $\beta^*_{k,n}$ as  the $n$th  departure times of the valid packet; that is, $\beta^*_{k,n}=\alpha_{k,n}^*+S_{k,n}^*$. Without loss of generality, we assume that
\[
\cdots<\beta^*_{k,0}\le 0<\beta^*_{k,1}<\beta^*_{k,2}<\cdots.
\]
The service time distribution of valid packets is obtained as follows
\begin{lem}
\label{lem:S}
The service time of a valid packet of source $k\in\bbK$ follows the exponential distribution having mean $1/(\lambda+\mu)$.
\end{lem}
\begin{proof}
Let  $X_a$ denote a random variable  following the exponential  distribution  with mean $1/a$ for $a>0$. Note that packets arrive according to the Poisson process with rate $\lambda$ if we ignore  types of sources. In addition, a packet is valid if no other packets arrive until its service is completed. Thus, it follows that, for $k\in\bbK$ and $n\in\bbZ$,
\[
\PP(S^*_{k,n}> x)=\PP(X_\mu> x|X_\lambda>X_\mu)
=\re^{-(\lambda+\mu)x},\quad x\ge 0.
\]
\qed
\end{proof}

\begin{figure}[!t]
\begin{center}
\includegraphics[clip,width=.94\linewidth]{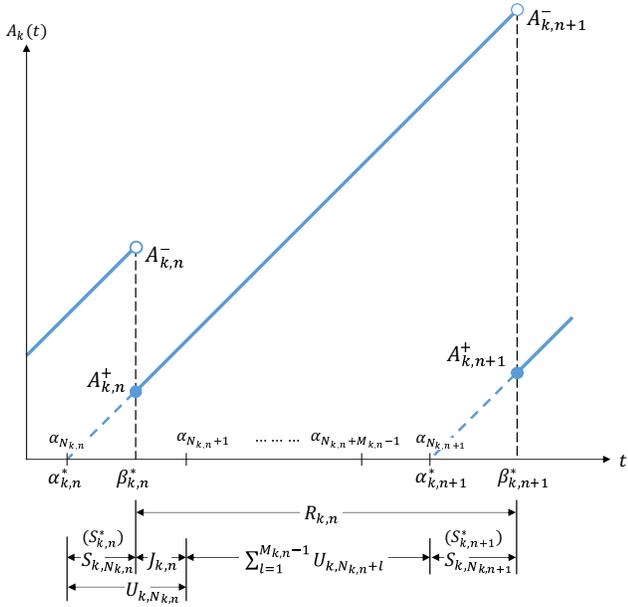}
\end{center}
\caption{A sample path of the AoI process $A_k(t)$ ($k\in\bbK$).}
\label{fig:AoI1}
\end{figure}

Using this lemma, we obtain the LST of each source.
\begin{thm}
\label{thm:moment}
The Laplace-Stieltjes transform of the stationary AoI of source $k\in\bbK$, denoted by $\wt{A}_k(s)$, is given by
\[
\wt{A}_k(s)={\lambda_k\mu\over (s+\lambda)(s+\mu)-(\lambda-\lambda_k)\mu},\qquad s\ge 0.
\]
\end{thm}
\begin{proof}
We define $A_{k,n}^+$  (resp. $A_{k,n}^-$) as the AoI of the immediately after (resp. before) the $n$th update of source $k$. Let  $\wt{A}_k^+(s)$ (resp. $\wt{A}_k^-(s)$) denote the LSTs of the stationary distributions of $\{A_{k,n}^+\}$ (resp. $\{A_{k,n}^-\}$). It follows from \cite{inou19} that, for $k\in\bbK$ and $s\ge0$,
\begin{equation}
\wt{A}_k(s)=\lambda_k^*{\wt{A}_k^+(s)-\wt{A}_k^-(s)\over s},
\label{eq:20200308-4}
\end{equation}
where $\lambda_k^*$ denotes the arrival rate of valid packets of source $k$.

Note here that $A_{k,n}^+$ is equivalent to the service time of the valid packet arriving at $\alpha_{k,n}^*$. Thus, from Lemma~\ref{lem:S}, we obtain 
\begin{equation}
\wt{A}_k^+(s)={\lambda+\mu\over s+\lambda+\mu }.
\label{eq:20200308-3}
\end{equation}
We  define $R_{k,n}$ as the interval time of $n$th and $(n+1)$st updates of source $k$; that is, $R_{k,n}=\beta^*_{k,n+1}-\beta^*_{k,n}$.  
The following relation holds for  $k\in\bbK$ and $n\in\bbZ$.
\begin{equation}
{A}_{k,n}^-={A}_{k,n-1}^++R_{k,n-1}.
\label{eq:20200308-5}
\end{equation}
We also have
\begin{equation}
\EE[\re^{-sR_{k,n}}]={\lambda_k\mu\over(s+\lambda)(s+\mu)-(\lambda-\lambda_k)\mu},
\label{eq:20200308-6}
\end{equation}
which is shown in Appendix~\ref{app:lem:R}.
Applying \eqref{eq:20200308-3} and \eqref{eq:20200308-6} to \eqref{eq:20200308-5} yields
\begin{equation}
\wt{A}_k^-(s)
={\lambda+\mu\over s+\lambda+\mu}{\lambda_k\mu\over(s+\lambda)(s+\mu)-(\lambda-\lambda_k)\mu}.
\label{eq:20200308-1}
\end{equation}
Furthermore, a packet of source $k\in\bbK$ is valid with probability $\mu/(\lambda+\mu)$, independently other packets. We then have
\begin{equation}
\lambda^*_k={\lambda_k\mu\over\lambda+\mu}.
\label{eq:20200308-2}
\end{equation}
Consequently, applying \eqref{eq:20200308-3}, \eqref{eq:20200308-1}, and \eqref{eq:20200308-2} to \eqref{eq:20200308-4}, we obtain Theorem~\ref{thm:moment}.\qed
\end{proof}

Furthermore, using Theorem~\ref{thm:moment}, we can easily obtain the expectation and variance of each AoI.
\begin{coro}
\label{coro:moment}
For $k\in\bbK$, we have
\[
\EE[A_k]
=
{1\over \lambda_k}\left\{1+{\lambda\over\mu}\right\},~~
\VV[A_k]
=
{1\over \lambda_k^2}\left\{1+2{\lambda-\lambda_k\over\mu}+{\lambda^2\over\mu^2}\right\}.
\]
\end{coro}
\medskip


\section{Correlation Coefficient}\label{sec:analysis2}

In~this~section,~assuming~that~$K=2$,~we~derive~the~correlation coefficient of stationary AoIs. 
For $n\in\bbZ$, let ${\alpha}_n$ denote the $n$th arrival time of packets of either sources 1 or 2, and let ${S}_n$ denote the service time of the packet arriving at~${\alpha}_n$.

In addition we define some notations related to valid packets.
We define ${\alpha}^*_n$ and  $S_n^*$ as the  arrival time and service time of the $n$th valid packet of either sources 1 or 2, respectively. 
We also define ${\beta}^*_n$ as  the $n$th  departure times of the valid packet of either sources 1 or 2; that is, ${\beta}^*_n={\alpha}_{n}^*+{S}_{n}^*$. Without loss of generality, we assume that
\[
\cdots<{\beta}^*_{0}\le 0<{\beta}^*_{1}<{\beta}^*_{2}<\cdots.
\]
We  define ${R}_{n}$ as the interval time of $n$th and $(n+1)$st updates; that is, ${R}_{n}={\beta}^*_{n+1}-{\beta}^*_{n}$.  We obtain the following lemma.
\begin{lem}
\label{lem:R} 
${R}_{n}$ follows the convolution of two independent random variables  following exponential distributions having mean $1/\lambda$ and $1/\mu$.
\end{lem}
The proof of Lemma~\ref{lem:R} is shown in Appendix~\ref{app:lem:R:2}.

We first consider the AoI of source $k\in\{1,2\}$ immediately after source 1 or 2 is updated.  We define $A_{k,n}^\dagger:=A_{k}(\beta_{n}^*)$ for $k=1,2$ and $n\in\bbZ$. We obtain the following lemma.
\begin{lem}
\label{lem:after Update}
$\{A_{1,n}^\dagger\}$, $\{A_{2,n}^\dagger\}$, and $\{A_{1,n}^\dagger A_{2,n}^\dagger\}$ are stationary and ergodic. In addition, we have, for $n\in\bbZ$,
\begin{eqnarray*}
\EE[A_{k,n}^\dagger]
&=&
{1\over \lambda+\mu}+{\lambda-\lambda_k\over\lambda_k}\left({1\over \lambda}+{1\over\mu}\right), \qquad k=1,2,
\\
\EE[A_{1,n}^\dagger A_{2,n}^\dagger]
&=&
2\left({1\over \lambda+\mu}\right)^2
\\
&&{}+\left({\lambda^2\over\lambda_1\lambda_2}-2\right)\left({1\over \lambda}+{1\over\mu}\right){1\over \lambda+\mu}.
\end{eqnarray*}
\end{lem}
\begin{proof}
For $n\in\bbZ$, let $C_n$ denote the type of the valid packet arriving at ${\alpha}_n^*$. For $k=1,2$ and $n\in\bbZ$, we define
\[
D_{k,n}:=\min\left\{m\le n;~C_m=k\right\},
\]
which means that the $D_{k,n}$th valid packet is the last packet which arrives from source $k$ before the $n$th update. Note that $D_{k,n}=n$ if the valid packet arriving at $\alpha_n^*$ is generated by source $k$. Using this notation,  we have
\begin{eqnarray}\textstyle
A_{k,n}^\dagger
=
{S}^*_{D_{k,n}}
+
\sum_{j=D_{k,n}}^{n-1}{R}_{j},
\label{eq:20200301-6}
\end{eqnarray}
which implies that $\{A_{1,n}^\dagger\}$, $\{A_{2,n}^\dagger\}$ and $\{A_{1,n}^\dagger A_{2,n}^\dagger\}$ are stationary and ergodicity.

Since all packets  have the same service time distribution, $\{C_n\}$ are the i.i.d. random variables such that $\PP(C_n=1)=\lambda_1/\lambda$ and $\PP(C_n=2)=\lambda_2/\lambda$. We then have
\begin{eqnarray}
\PP(D_{k,n}=m)
&=&
\left\{\begin{array}{ll}
(\lambda_k/ \lambda)(1-{\lambda_k/\lambda})^{n-m},
& m\le n,\\
0,&m>n.
\end{array}\right.
\label{eq:20200301-5}
\end{eqnarray}
Note that $\{D_{k,n}\}$ is independent of $\{S^*_n\}$ and $\{R_n\}$. It then follows from \eqref{eq:20200301-6} and \eqref{eq:20200301-5} that, for $n\in\bbZ$ and $k=1,2$,
\begin{eqnarray*}
\EE[A_{k,n}^\dagger]
&=&
\textstyle\EE\left[\EE\left[{S}^*_{D_{k,n}} +
\sum_{j=D_{k,n}}^{n-1}
{R}_j\,\big|\, D_{k,n}\right]\right]\\
&=&\EE[S_0]+{\lambda\over\lambda_k}\left(1-{\lambda\over\lambda_k}\right)\EE[R_0].
\end{eqnarray*}
Applying Lemmas~\ref{lem:S} and \ref{lem:R} to the above, we obtain
\[
\EE[A_{k,n}^\dagger]
={1\over \lambda+\mu}+{\lambda\over\lambda_k}\left(1-{\lambda_k\over\lambda}\right)\left({1\over \lambda}+{1\over\mu}\right).
\]

\begin{figure}[!t]
\begin{center}
\includegraphics[clip,width=.90\linewidth]{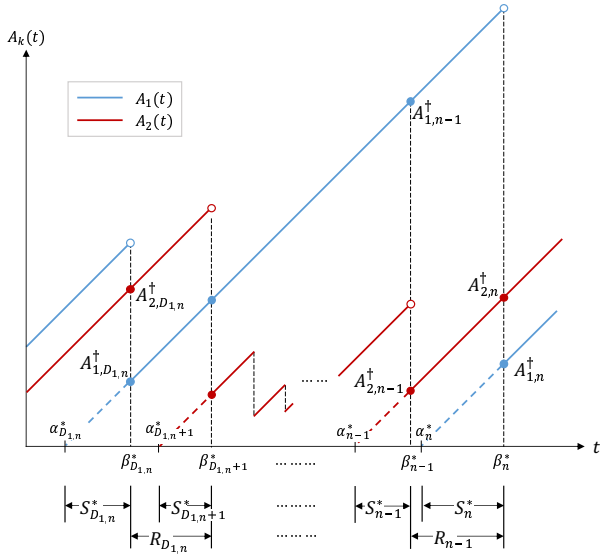}
\end{center}
\caption{Two sample paths of the AoI processes $A_k(t)$ ($k\!=\!1,2$).}
\label{fig:AoI2}
\end{figure}

In addition, we have, for $m,\ell\in\bbZ$,
\begin{eqnarray}
\lefteqn{
\PP(D_{1,n}=m,D_{2,n}=\ell)
}~~~~&&\nonumber\\
~&=&
\left\{\begin{array}{ll}
({\lambda_2/\lambda})({\lambda_1/ \lambda})^{n-\ell},
& m=n,~\ell<n,\\
({\lambda_1/ \lambda})({\lambda_2/ \lambda})^{n-m},
& m<n,~\ell=n,\\
0,&otherwise.
\end{array}\right.
\label{eq:20200301-7}
\end{eqnarray}
Note that  $S_n^*$ does not depend on $S_{n-m}^*$ and $R_\ell$ for $\ell\neq n-1$, but depends on $R_{n-1}$. Thus, from \eqref{eq:20200301-6} and \eqref{eq:20200301-7}, we obtain
\begin{eqnarray*}
\lefteqn{
\EE[A_{1,n}^\dagger A_{2,n}^\dagger]
}{}&&\\
&&{}=
\EE\left[\EE[A_{1,n}^\dagger A_{2,n}^\dagger|D_{1,n},D_{2,n}]\right]
\\
&&{}=
\EE[{S}^*_0]^2+
\left\{{\lambda^2\over\lambda_1\lambda_2}-3\right\}
\EE\left[{S}^*_0\right]\EE\left[{R}_0\right]+\EE\left[{S}^*_1{R}_{0}\right],
\end{eqnarray*}
By definitions of $R_n$, we have
\[
S^*_1R_0=S^*_1\left\{(\alpha_1^*+S^*_1)-(\alpha_0^*+S^*_0)\right\}.
\]
Thus, it follows from Lemma~\ref{lem:S} and \eqref{eq:20200308-2} that
\begin{eqnarray}
\EE[S^*_{n}R_{n-1}]
&=&
\EE[S^*_1]\EE[\alpha_1^*-\alpha_0^*]+\EE[(S^*_1)^2]-\EE[S^*_1]\EE[S^*_0]
\nonumber\\
&=&
\left({1\over\lambda}+{1\over\mu}\right)\left({1\over\lambda+\mu}\right)+\left({1\over\lambda+\mu}\right)^2.
\label{eq:20200407-1}
\end{eqnarray}
Applying Lemmas~1,~2 and \eqref{eq:20200407-1} to the above, we obtain
\begin{eqnarray*}
\EE[A_{1,n}^\dagger A_{2,n}^\dagger]
&=&
2\left({1\over \lambda+\mu}\right)^2
\\
&&{}+\left({\lambda^2\over\lambda_1\lambda_2}-2\right)\left({1\over \lambda}+{1\over\mu}\right){1\over \lambda+\mu}.
\end{eqnarray*}
\qed
\end{proof}

Using Lemma~\ref{lem:after Update}, we obtain the main theorem of this paper.
\begin{thm}
\label{thm:coefficient}
The correlation coefficient of AoIs in 2-source M/M/1/1 push-out queue, denoted by $\rho$, is given by
\begin{equation}
	\rho={-2\lambda_1\lambda_2\mu \over \lambda \sqrt{(\lambda^2+2\lambda_1\mu+\mu^2)(\lambda^2+2\lambda_2\mu+\mu^2)}}
\end{equation}
\end{thm}

\begin{proof}
Using the pointwise ergodic theorem (see, e.g., \cite[Theorem 1.6.4]{Bacc13}) yields
\begin{equation}
\EE[A_1A_2]
=\lim_{T\to\infty}{1\over T}\int_0^TA_1(t)A_2(t)\rd t.
\label{defn:ol_A1A2}
\end{equation}
Dividing the integral in \eqref{defn:ol_A1A2}  by update times $\{\beta_n\}$, we have
\begin{eqnarray}
\EE[A_1A_2]
&=&
\lim_{T\to\infty}{\Delta(T)}
+
\lim_{T\to\infty}{N(T)\over T}\cdot {1\over N(T)}
\sum_{n=1}^{N(T)}F_n,\qquad
\label{eq:20200307-4}
\end{eqnarray}
where  $N(t)$ denotes the total number of updates in $[0,t)$ and
\begin{eqnarray}
\delta(T)&=&
-{1\over T}\int_{\beta^*_0}^{0}A_1(t)A_2(t)\rd t+{1\over T}\int_{\beta^*_{N(T)}}^{T}A_1(t)A_2(t)\rd t,\nonumber\\F_n&=&
\int_{\beta^*_{n-1}}^{\beta^*_n}A_1(t)A_2(t)\rd t.
\label{eq:20200306-1}
\end{eqnarray}
Note that, for $t\in[\beta^*_{n},\beta^*_{n+1})$, 
\begin{eqnarray}
A_1(t)A_2(t)
&=&\left(t-\beta^*_{n}+A_{1,n}^\dagger\right)\left(t-\beta^*_{n}+A_{2,n}^\dagger\right).\quad
\label{eq:20200306-2}
\end{eqnarray}

We estimate the right-hand side of \eqref{eq:20200307-4}. From \eqref{eq:20200306-2}, we have
\begin{eqnarray}
\left|\Delta(T)\cdot T\right|&&\le
{R_{0}}(A_{1,0}^\dagger+R_{0})(A_{2,0}^\dagger+R_0)
\label{eq:20200307-5}\\
&&{}{}{}+{R_{N(T)}}(A^\dagger_{1,N(T)}+R_{N(T)})(A^\dagger_{2,N(T)}+R_{N(T)}).
\nonumber
\end{eqnarray}
It follows from Lemmas~\ref{lem:R} and \ref{lem:after Update} that $R_n$ and $A_k({\beta}^*_{n})$ are finite w.p.1. Thus, it follows from \eqref{eq:20200307-5} that 
\begin{eqnarray}
\lim_{T\to\infty}{\Delta(T)}=0.
\label{eq:20200307-1}
\end{eqnarray}
Furthermore, $\{R_n\}$ is the i.i.d.\ random variables,  because the system becomes empty at time $\beta_n^*$, $n\in\bbZ$. Thus,  it follows from the  elementary renewal theorem and Lemma~\ref{lem:R} that 
\begin{eqnarray}
\lim_{T\to\infty}{N(T)\over T}
={1\over \EE[R_0]}=\left({1\over\lambda}+{1\over\mu}\right)^{-1}.
\label{eq:20200229-2}
\end{eqnarray}
Furthermore, using the pointwise ergodic theorem, we have
\begin{eqnarray}
\lim_{T\to\infty}{1\over N(T)}\sum_{n=1}^{N(T)}F_n
=
\lim_{N\to\infty}{1\over N}\sum_{n=1}^{N}F_n
=\EE\left[F_0\right].
\label{eq:20200306-3}
\end{eqnarray}
where the first equation holds because it follows from \eqref{eq:20200229-2} that $\lim_{T\to\infty}N(T)=\infty$. 
Substituting \eqref{eq:20200307-1}--\eqref{eq:20200306-3} into \eqref{eq:20200307-4} yields
\begin{eqnarray}
\EE[A_1A_2]
&=&
\left({1\over\lambda}+{1\over\mu}\right)^{-1}\cdot\EE\left[F_0\right].
\label{eq:20200307-6}
\end{eqnarray}

Next, we calculate $\EE[F_0]$. Applying \eqref{eq:20200306-2} to \eqref{eq:20200306-1}, we have
\begin{eqnarray}
F_n
&=&
{R_n^3\over 3}
+
{R_n^2\over 2}(A_{1,n}^\dagger+A_{2,n}^\dagger)
+
R_nA_{1,n}^\dagger A_{2,n}^\dagger.\nonumber
\end{eqnarray}
Using Lemmas~\ref{lem:R} and \ref{lem:after Update}, it follows from the above that
\begin{eqnarray*}
\EE\left[F_0\right]
&=&
{\EE[R_0^3]\over 3}
+
{\EE[R_0^2]\over 2}\EE[A_{1,n}^\dagger+A_{2,n}^\dagger]
+
\EE[A_{1,n}^\dagger A_{2,n}^\dagger]\EE[R_0^*]
\\
&=&
{\lambda^2\over\lambda_1\lambda_2}\left({1\over \lambda}+{1\over\mu}\right)^3-2\left({1\over \lambda}+{1\over\mu}\right)^2{1\over \lambda+\mu}.
\label{eq:20200301-14}
\end{eqnarray*}
Combining the above and  \eqref{eq:20200307-6}, we obtain
\[
\EE[A_1A_2]
=
{\lambda^2\over\lambda_1\lambda_2}\left({1\over \lambda}+{1\over\mu}\right)^2-2\left({1\over \lambda}+{1\over\mu}\right){1\over \lambda+\mu}.
\]

Consequently, from  the above and Corollary~\ref{coro:moment}, we obtain
\begin{eqnarray*}
\rho
&=&
{\EE[A_1A_2]-\EE[A_1]\EE[A_2]\over \sqrt{\VV[A_1]\VV[A_2]}}
\\
&=&
{-2\lambda_1\lambda_2\mu \over \lambda \sqrt{(\lambda^2+2\lambda_1\mu+\mu^2)(\lambda^2+2\lambda_2\mu+\mu^2)}}.
\end{eqnarray*}
\qed
\end{proof}

From Theorem~\ref{thm:coefficient}, we can see some nontrivial properties on the systems as follows. The proof of Corollary~\ref{cor:property} is omitted due to the page restriction.
\begin{coro}\label{cor:property}
The following statements holds.
\begin{itemize}
\item[(i)]
AoIs of sources 1 and 2 have a negative correlation; that is, $\rho<0$.
\item[(ii)]
When any of $\lambda_1$, $\lambda_2$, or $\mu$ approaches infinity, the correlation coefficient $\rho$ converges to zero.
\item[(iii)]
The minimum value of the correlation coefficient $\rho$ is $-1/6$, which is achieved when $\lambda_1/2=\lambda_2/2=\mu$.
\end{itemize}
\end{coro}
\medskip


\section{Numerical Results}\label{sec:evaluation}

In this section, we provide numerical results of the correlation coefficient presented in Theorem~\ref{thm:coefficient}, and confirm the statements presented in Corollary~\ref{cor:property} through the numerical results.
Fig.~\ref{fig:numerical} shows the correlation coefficients which is numerically computed with several cases using Theorem~\ref{thm:coefficient}, where the $x$-axis represents the value of  $\lambda_1$, the arrival rate of source~1. Note that the parameter $\mu$ is fixed for each curve in Fig.~\ref{fig:numerical}, and $\lambda_2$ is fixed as $\lambda_2=2$.

From Fig.~\ref{fig:numerical}, we observe that $\rho$ is always negative, that is, the two AoIs in our model always have a negative correlation, which implies (i) in Corollary~\ref{cor:property}. We also find that the curves in the figures are all convergent to zero, so that (ii) in Corollary~\ref{cor:property} is likely to hold with respect to $\lambda_1$. Moreover, we see that the correlation coefficient has a minimum value with respect to $\lambda_1$. This means that a certain arrival rate of packets from one source gives the strongest negative correlation with the other source. Furthermore, we see from the figure that the smallest minimum value of $\rho$ is seen when $\mu=4$.
Actually, the smallest value in the figure is $-1/6$, and is achieved when $(\lambda_1, \lambda_2, \mu)=(2,2,4)$, which corresponds to (iii) in Corollary~\ref{cor:property}.

\begin{figure}[!t]
\begin{center}
\includegraphics[clip,width=.94\linewidth]{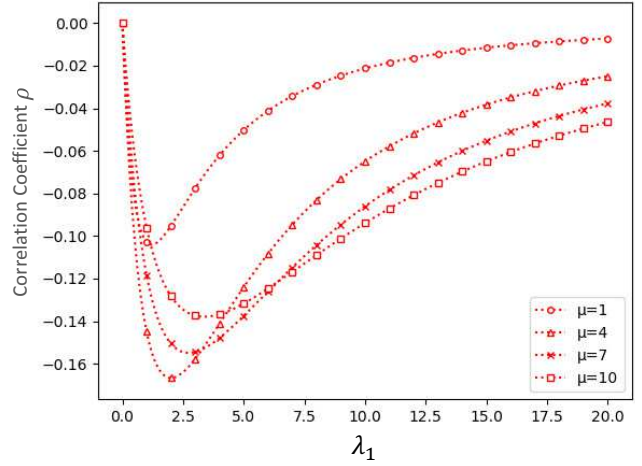}
\end{center}
\caption{Numerically computed results of $\rho$ with several patterns of the parameter $\mu$, and with fixed parameter $\lambda_2$ as $\lambda_2=2$.}
\label{fig:numerical}
\end{figure}
\medskip


\section{Conclusion}\label{sec:conclusion}

In this paper, we considered a correlation of the AoIs of two sources sharing one server to process information, and derived the closed form expression of the correlation coefficient, on the model of M/M/1/1 queueing systems. In addition, we also derived the expression of the LST of the stationary distribution of each AoI on the assumption that multiple $K$ sources share the one server. From our analysis, we found that the correlation coefficient is always negative, and that the correlation coefficient has a certain minimum value. This indicates that there is always a negative correlation between the two AoIs, and the strongest negative correlation is achieved by adjusting the parameters introduced in our model.

For further study, it would be expected that the correlation coefficient is investigated with more generalized assumptions for modeling the systems because our model adopts an elemental M/M/1/1 queueing systems with the common service rate for all the sources. To consider the correlation among several sources is also an interesting research point.
\medskip

\appendices 
\section{}
\label{app:lem:R}
This appendix is devoted to the proof of \eqref{eq:20200308-6}.
For $k\in\bbK$ and $n\in\bbZ$, let $N_{k,n}$ denote the number of arriving packets of source $k$ in $(0,\beta_n^*]$. 
We define $M_{k,n}$ as  the number of arriving packets of source $k$ in $(\beta_n^*,\beta_{n+1}^*]$; that is, $M_{k,n}=N_{k,n+1}-N_{k,n}$. In addition, we define $U_{k,n}:=\alpha_{k,n+1}-\alpha_{k,n}$.

Let also ${J}_{k,n}$ denote the length from the time that the $n$th valid packet departs to the time that a new packet arrives from source $k$; that is $J_{k,n}=U_{k,N_{k,n}}-S_{k,N_{k,n}}$. 
We have the following relation (see Fig.~\ref{fig:AoI1}).
\[
R_{k,n}=J_{k,n}+\sum_{\ell=1}^{M_{k,n}-1}U_{k,N_{k,n}+\ell}
+S_{k,N_{k,n+1}}.
\]
which leads to
\begin{eqnarray}
\EE\left[\re^{-sR_{k,n}} \big|M_{k,n}\right] 
&=&
\EE\left[\re^{-sJ_{k,n} }\right]\cdot\EE\left[\re^{-sS_{k,N_{k,n+1}}}\right]
\nonumber\\
&&{}{}\times{{\textstyle\prod_{\ell=1}^{M_{k,n}-1}} }\EE\left[\re^{-sU_{k,N_{k,n}+\ell}}\big|M_{k,n}\right].\quad~~
\label{20200301-2-2}
\end{eqnarray}
For $\ell=1,\dots,M_{k,n}-1$, the $(N_{k,n}+\ell)$-th arrival packet is not valid. 
We then have, for $\ell=1,\dots,M_{k,n}-1$,

\begin{eqnarray}
\EE\left[\re^{-sU_{N_{k,n}+\ell}}\big|M_{k,n}\right]
&=&\EE\left[\re^{-sX_{\lambda_k}}\big|{\textstyle\min_{\ell\in\bbK}}X_{\lambda_\ell} \le X_\mu\right]
\nonumber\\
&=&
{\lambda+\mu\over s+\lambda+\mu }{s+\lambda\over \lambda}{\lambda_k\over s+\lambda_k},
\label{eq:20202029-6-2}
\end{eqnarray}
where  $X_a$ denotes a random variable  following the exponential  distribution  with mean $1/a$ for $a>0$.
Furthermore, it follows from the memoryless property that 
\begin{eqnarray}
\EE\left[\re^{- sJ_{k,n}}\right]
&=&\EE\left[\re^{-s(X_{\lambda_k}-X_{\mu})} \big|X_{\lambda_k}>X_\mu\right]
=
{\lambda_k\over s+\lambda_k}.\qquad
\end{eqnarray}
Note that $S_{k,N_{k,n+1}}$ is the service time of a valid packet. Thus, it follows from Lemma~\ref{lem:S} that 
\begin{eqnarray}
\EE[\re^{-sS_{k,N_{k,n+1}}}]={\lambda+\mu\over s+\lambda+\mu}.
\label{eq:20202029-7-2}
\end{eqnarray}
Applying \eqref{eq:20202029-6-2}--\eqref{eq:20202029-7-2} to \eqref{20200301-2-2}, we obtain 
\begin{eqnarray}
\EE\left[\re^{-sR_{k,n}}|M_{k,n}\right]
&=&{\lambda+\mu\over s+\lambda+\mu}{\lambda_k\over s+\lambda_k}
\nonumber\\
&&{}\times\left({\lambda+\mu\over s+\lambda+\mu}{\lambda_k\over s+\lambda_k}{s+\lambda\over \lambda}\right)^{M_{k,n}-1}.
\label{2-2--319-1}
\end{eqnarray}

Finally, we show the distribution function of $M_{k,n}$. A packet is valid if no packets (of all sources) arrive before its service completion.
It then follows from the independence of $\{{\alpha}_{k,n}\}$ and $\{{S}_{k,n}\}$ that a packet is valid with probability $\mu/(\lambda+\mu)$, independently other packets. Therefore, for any $n\in\bbN$ and $k\in\bbK$, $M_{k,n}$ follows the geometric distribution on $\bbN$ with parameter $\mu/(\lambda+\mu)$. Thus, from \eqref{2-2--319-1}, we obtain%
\begin{eqnarray}
\EE[\re^{-sR_{k,n}}]
&=&
{\lambda_k\mu\over(s+\lambda)(s+\mu)-(\lambda-\lambda_k)\mu}.
\label{app:20200308-1}
\end{eqnarray}
\medskip

\section{}
\label{app:lem:R:2}
This appendix is devoted to the proof of Lemma~\ref{lem:R}.
We derive the moment generation function of ${R}_{n}$. Let ${N}_n$ denote the total number of packets arriving from either source 1 and 2 in $(0,{\beta}_n^*]$.  In addition, we define $U_{n}:=\alpha_{n+1}-\alpha_{n}$.

We define ${M}_n:={N}_{n+1}-{N}_n$ and ${J}_{n}:={U}_{{N}_{n}}-{S}_{{N}_{n}}$. As similar way to \eqref{20200301-2-2}, we have
\begin{eqnarray}
\EE\left[\re^{-s{R}_n} \big|M_n\right]
&=&
\EE\left[\re^{-s{J}_n }\right]\cdot\EE\left[\re^{-s{S}_{{N}_{n+1}}}\right]
\nonumber\\
&&{}{}\times{\textstyle\prod_{\ell=1}^{M_n-1}} \EE\left[\re^{-s{U}_{{N}_n+\ell}}\big|{M}_n\right].
\label{20200301-1}
\end{eqnarray}

For $\ell=1,\dots,M_n-1$, the $(N_n+\ell)$-th arrival packet is not valid. 
We then have, for $\ell=1,\dots,M_n-1$,
\begin{eqnarray}
\EE\left[\re^{-s{U}_{{N}_n+\ell}}\big|{M}_n\right]
&=&\EE\left[\re^{-sX_\lambda}|X_\lambda \le X_\mu\right]
=
{\lambda+\mu\over s+\lambda+\mu}.\qquad
\label{eq:20202029-6}
\end{eqnarray}
Furthermore, it follows from the memoryless property that 
\begin{eqnarray}
\EE\left[\re^{- s{J}_n}\right]
&=&\EE\left[\re^{-s(X_\lambda-X_\mu)} \big|X_\lambda>X_\mu\right]
=
{\lambda\over s+\lambda}.
\label{eq:20202029-7}
\end{eqnarray}
Note that ${S}_{{N}_{n+1}}$ is the service time of a valid packet. Thus, it follows from Lemma~\ref{lem:S} that 
\begin{eqnarray}
\EE\left[\re^{-s{S}_{{N}_{n+1}}}\right]={\lambda+\mu\over s+\lambda+\mu}.
\label{eq:20200308-10}
\end{eqnarray}
Applying \eqref{eq:20202029-6}--\eqref{eq:20200308-10} to \eqref{20200301-1}, we obtain
\begin{eqnarray}
\EE\left[\re^{-s{R}_n} \big|M_n\right] 
&=&\left({\lambda+\mu\over s+\lambda+\mu}\right)^{{M}_n}{\lambda\over s+\lambda} .
\label{20200301-2}
\end{eqnarray}

As similar  to \eqref{20200301-2-2}, for any $n\in\bbN$, $M_n$ follows the geometric distribution on $\bbN$ with parameter $\mu/(\lambda+\mu)$. Thus, from \eqref{20200301-2}, we obtain
\begin{eqnarray*}
\EE[\re^{-s{R}_n}]
={\lambda\over s+ \lambda}{\mu\over s+\mu},
\end{eqnarray*}
which means that ${R}_n$ follows  the convolution of $X_\lambda$ and $X_\mu$.
\medskip



\begin{thebibliography}{00}

\bibitem{KaulYatesGrut12}
S. Kaul, R. Yates, and M. Gruteser,
\newblock {Real-time status: How often should one update?}.
in~\emph{Proc}.
\newblock {\it IEEE INFOCOM},
p.~2731–-2735, 2012.


\bibitem{YatesKaul12}
R. Yates and S. Kaul, 
\newblock {Real-time status updating: Multiple sources}.
in~\emph{Proc}.
\newblock {\it IEEE Int. Symp. on Inf. Theory},
p.~2666--2670, 2012.

\bibitem{KaulYates12}
S. Kaul, R. Yates, and M. Gruteser,
\newblock {Status updates through queues}.
in~\emph{Proc}.
\newblock {\it IEEE CISS},
p.~1--6, 2012.

\bibitem{NajmTela18}
E. Najm and E. Telatar,
\newblock {Status updates in a multi-stream M/G/1/1 preemptive queue}.
in~\emph{Proc}.
\newblock {\it IEEE INFOCOM Workshops},
p.~124--129, 2018.

\bibitem{YatesKaul18}
R. Yates and S. Kaul,
\newblock {The age of information: Real-time status updating by multiple sources}.
\newblock {\it IEEE Trans. on Inf. Theory},
vol.~65, no.~3, p.~1807--1827, 2019.

\bibitem{Bacc13}
F. Baccelli and P. Br{\'e}maud,
\newblock {\it Elements of queueing theory: Palm Martingale calculus and stochastic recurrences}.
\newblock Springer Science \& Business Media,
2013.

\bibitem{inou19}
Y. Inoue, and H. Masuyama, T. Takine and T. Tanaka,
\newblock {A general formula for the stationary distribution of the age of information and its application to single-server queues}.
\newblock {\it IEEE Trans. on Inf. Theory},
vol.~65, no.~12, p.~8305--8324, 2019.
\end{thebibliography}
\end{document}